\documentclass[a4paper,10pt,oneside]{article}

\usepackage{amsmath}
\usepackage{amssymb}
\usepackage{amsthm}

\newcommand*{\N}{\mathbb{N}}                           
\newcommand*{\R}{\mathbb{R}}                           
\newcommand*{\Rp}{\mathbb{R}^{+}}                      

\newcommand*{\MN}[1]{\mathrm{M}_{#1}}                  
\newcommand*{\MNp}[1]{\mathrm{M}_{#1}^{+}}             
\newcommand*{\TMN}[1]{\mathrm{M}_{#1,\mathrm{sa}}^{0}} 
\newcommand*{\MNsa}[1]{\mathrm{M}_{#1,\mathrm{sa}}}    
\newcommand*{\MNN}[1]{\mathcal{M}_{#1}}                

\newcommand*{\di}{\displaystyle}                       

\newcommand*{\wt}[1]{\widetilde{#1}}                   
\newcommand*{\ce}{\mathrm{e}}                          
\newcommand*{\ci}{\mathrm{i}}                          
\newcommand*{\Tr}{\mathrm{Tr\ }}                       
\newcommand*{\Diag}{\mathrm{Diag}}                     
\newcommand*{\clr}{\mathrm{clr}}                       
\newcommand*{\Lm}[1]{L_{#1}}                           
\newcommand*{\Rm}[1]{R_{#1}}                           
\DeclareMathOperator\artanh{artanh}                    

\newcommand*{\gz}[1]{\left(#1\right)}                  
\newcommand*{\sz}[1]{\left[#1\right]}                  
\newcommand*{\kz}[1]{\left\{#1\right\}}                
\newcommand*{\norm}[1]{\left\Vert #1\right\Vert}       
\newcommand*{\scal}[1]{\left\langle#1\right\rangle}    
\newcommand*{\add}{\oplus}                             
\newcommand*{\mul}{\odot}                              

\makeatletter
\let\@afterindentfalse\@afterindenttrue
\@afterindenttrue
\makeatother    

\newtheorem{definition}{Definition}
\newtheorem{theorem}{Theorem}
\newtheorem{lemma}{Lemma}

\title{Quantum Aitchison geometry\thanks{
keywords: quantum state, quantum information, information geometry;
MSC2010: 81P16, 81P45.
}}

\author{
Attila Andai\thanks{andatt@gmai.com},
Attila Lovas\thanks{lovas@math.bme.hu}\\
Department of Mathematical Analysis, \\
Budapest University of Technology and Economics,\\
Egry J\'ozsef u. 1, Budapest, 1111, Hungary
}

\date{\today}

\begin{document}

\maketitle

\begin{abstract}
Multiplying a likelihood function with a positive number makes no difference in Bayesian statistical 
  inference, therefore after normalization the likelihood function in many cases can be considered as probability 
  distribution.
This idea led Aitchison to define a vector space structure on the probability simplex in 1986.
Pawlowsky-Glahn and Egozcue gave a statistically relevant scalar product on this space in 2001, endowing the
  probability simplex with a Hilbert space structure.
In this paper we present the noncommutative counterpart of this geometry.
We introduce a real Hilbert space structure on the quantum mechanical finite dimensional state space.
We show that the scalar product in quantum setting respects the tensor product structure
  and can be expressed in terms of modular operators and Hamilton operators.
Using the quantum analogue of the log-ratio transformation it turns out that all the newly introduced operations emerge
  naturally in the language of Gibbs states.
We show an orthonormal basis in the state space and study the introduced geometry on the space of qubits in details.
\end{abstract}

\section{Introduction}

The very first step towards the Euclidean structure on the quantum mechanical state
  space was done by Aitchison in 1986 when he noticed that 
  the operation of Bayes's formula to change a prior probability assessment $\pi$ into a posterior one $p$
  through a likelihood function $\rho$, can be viewed as an operation on the probability simplex 
  \cite{AitchisonDataAnalysisFirst,AitchisonOrig,AitchisonRoleOfPert}.
Multiplying a likelihood function with a positive constant results the same updated
  probability distribution, therefore in many cases with an appropriate multiplicative factor 
  one can transform the likelihood to a probability distribution $\rho$.
Such ideas led to the definition of an abstract operation $\oplus$ for the updating process
  $p=\rho\oplus\pi$.
The operation $\oplus$ is called perturbation, because in applications the likelihood function 
  has just a slight effect on the prior distribution.

To be more concrete let us denote the interior of the probability simplex by $S_{n}$ ($n\in\N\setminus\kz{0,1}$), that is
\begin{equation*}
 S_{n}=\kz{(p_{1},\dots,p_{n})\in\left]0,1\right[^{n}\Bigm\vert\ \sum_{i=1}^{n}p_{i}=1 }.
\end{equation*}
The perturbation is defined as an $S_{n}\times S_{n}\to S_{n}$ operation, if $p,q\in S_{n}$ then
\begin{equation}
\label{eq:oplusCl}
 p\oplus q=\frac{1}{\di\sum_{i=1}^{n}p_{i}q_{i}}\times (p_{1}q_{1},\dots,p_{n}q_{n})
\end{equation}
  and scalar multiplication is a $\R\times S_{n}\to S_{n}$ map, for $\lambda\in\R$ and $p\in S_{n}$
\begin{equation}
\label{eq:omulCl}
\lambda\mul p=\frac{1}{\di\sum_{i=1}^{n}p_{i}^{\lambda}}\times (p_{1}^{\lambda},\dots,p_{n}^{\lambda}).
\end{equation}
It is easy to check that $(S_{n},\oplus,\mul)$ is a vector space, where the zero vector is the uniform distribution.
This linear structure of the simplex was studied in details \cite{AitchisonAlgMethods,Billheimer}.
Note that intervals in this Aitchison geometry $(t\mul p)\oplus((1-t)\mul q)$ ($t\in\sz{0,1}$) play crucial
  role in hypothesis testing and have been studied in detail for decades in connection with relative entropy and 
  Chernoff distance \cite{InfoBookCover,InfoBookCsiszar,PetzQinf}.   
In statistic these intervals are often referred to as Hellinger arcs.
  
In 2001, Pawlowsky-Glahn and Egozcue using the notion of metric center and metric variation endowed 
  this space with an inner product 
\begin{equation}
\label{eq:oscalCl}
\scal{\cdot,\cdot}_{\circ}:S_{n}\times S_{n}\to\R \qquad
(p,q)\mapsto\frac{1}{2n}\sum_{i,j=1}^{n}\log\gz{\frac{p_{i}}{p_{j}}}\log\gz{\frac{q_{i}}{q_{j}}}
\end{equation}
  which turned out to be a simple tool to prove essential properties of statistical inference 
  \cite{EgozcueGlahnFirst,PawlowskyBLU}.
An orthonormal basis was constructed in the Hilbert space $(S_{n},\scal{\cdot,\cdot}_{\circ})$  which has particular   
  importance in estimation theory \cite{AitchisonONB}.
The norm $\norm{\cdot}_{\circ}$ induced by the scalar product is called information evidence $I_{e}$, since it fulfills
  the natural requirements that from different prior distributions $\pi_{1},\pi_{2}$ the same measurement
  inducing likelihood $\rho$, the information gain in posterior distributions $p_{i}=\rho\oplus\pi_{i}$ ($i=1,2$)
  should be the same $I_{e}(p_{1}\ominus\pi_{1})=I_{e}(p_{2}\ominus\pi_{2})=I_{e}(\rho)$
  and the distance between the prior and posterior distributions should be the same
  $I_{e}(p_{1}\ominus p_{2})=I_{e}(\pi_{1}\ominus\pi_{2})$ \cite{EgozcueEvidence}.
The formal definition of $p\ominus\pi$ is $p\oplus((-1)\mul\pi)$.
  
This geometrical structure of the simplex was extended to different statistical models recently.
For probability density functions on bounded intervals of the real numbers \cite{BarreroHilbertSpace},
  for $\sigma$-finite measures on a measurable space \cite{BoogaartBayesLinSpaces,BoogaartBayesHilbert}.
In this paper a similar Hilbert space structure is presented on the quantum mechanical state space 
  inspired by the classical Aitchison geometry.

\section{Hilbert space structure on the quantum mechanical state space}
\subsection{Basic notations and lemmas}

We work in finite dimensional framework, that is every Hilbert space will be 
  assumed to be finite dimensional over the complex field.
Let us fix some notation.
The letters $n,m$ denote integers greater than or equal to two.
The symbol $\MN{n}$ stands for the algebra of $n\times n$ complex matrices,
  $\MNsa{n}$ denotes the set of self-adjoint elements in $\MN{n}$ and
  $\MNp{n}$  denotes the set of positive definite matrices in $\MNsa{n}$.
The trace one elements of $\MNp{n}$ form the interior of the $n$-level quantum mechanical state
  space which is denoted by $\MNN{n}$, that is $\MNN{n} = \kz{D\in\MNp{n} |\Tr D = 1}$.
The linear structure of self-adjoint traceless matrices will play role in computations therefore
  we introduce the abbreviation $\TMN{n}$ for the set of those matrices.
The symbol $I_{n}$ is used for the $n\times n$ identity matrix and
  to shorten formulae for every matrix $A\in\MNp{n}$ the abbreviation $\wt{A}=\log(A)$ is used.
The matrix units $E_{ij}$ will be used, $E_{ij}$ is the matrix with every component zero,
  except the component $ij$, which is one and the size of the matrix units will be clear from the context.  

To a matrix $X\in\MN{n}$, one can associate the left and right multiplication 
  operators $\Lm{X},\Rm{X}:\MN{n}\to\MN{n}$ that act like
\begin{align*}
A \mapsto \Lm{X} (A) &= XA \\
A \mapsto \Rm{X} (A) &= AX. 
\end{align*}
For states $D_{1},D_{2}\in\MNN{n}$ the relative modular operator is defined as
  $\Delta_{D_{1}/D_{2}}=\Lm{D_{1}}\Rm{D_{2}^{-1}}$, 
  that is for every $A\in\MN{n}$ one has $\Delta_{D_{1}/D_{2}}(A)=D_{1}AD_{2}^{-1}$.
In the case $D_{1}=D_{2}$ the short notation $\Delta_{D}$ is used for $\Delta_{D/D}$.

The space $\MN{n}$ is endowed with  Hilbert-Schmidt inner product
\begin{equation*}
\scal{A,B}_{n}=\Tr A^{*}B\quad\forall A,B\in \MN{n}.
\end{equation*}
For every matrix $X\in\MNsa{n}$ the multiplication operators 
 $\Lm{X}$, $\Rm{X}$ and in the invertible case $\Delta_{X}$ too 
 are self-adjoint, moreover, for every $D\in\MNN{n}$ the operators $\Lm{D}$, $\Rm{D}$ and $\Delta_{D}$ are
 positive definite.

\begin{lemma}
\label{th:tensorlog}
For matrices $A\in\MNp{n}$ and $B\in\MNp{m}$ one has for their logarithm
\begin{equation*}
\wt{A\otimes B}=\wt{A}\otimes I_{m}+I_{n}\otimes\wt{B}
\quad \mbox{and}\quad
\wt{A\otimes A^{-1}}=\wt{A}\otimes I_{n}-I_{n}\otimes\wt{A}.
\end{equation*}
\end{lemma}
\begin{proof}
The tensor product $A\otimes B$ identified as an $\MN{m}\to\MN{n}$ linear map acting like $X\mapsto AXB^{*}$.
Using the commutativity of the operators $\Lm{A}$ and $\Rm{B}$ 
\begin{equation*}
\wt{A\otimes B}=\wt{\Lm{A}\Rm{B}}=\wt{\Lm{A}}+\wt{\Rm{B}}
  =\Lm{\wt{A}}+\Rm{\wt{B}}=\wt{A}\otimes I_{m}+I_{n}\otimes\wt{B}
\end{equation*}
follows.
\end{proof}

\subsection{Elementary operations}

Now we are in the position to present the quantum analogue of the perturbation $\oplus$,
  multiplication $\mul$ and scalar product $\scal{\cdot,\cdot}_{\circ}$ given by
  Equations (\ref{eq:oplusCl},\ref{eq:omulCl},\ref{eq:oscalCl}).

\begin{definition}
On the set of positive definite matrices $\MNp{n}$ define the following operations for every
  $A,B\in\MNp{n}$ and $\lambda\in\R$.
\begin{align*}
A\add B&:=\frac{\ce^{\wt{A}+\wt{B}}}{\Tr \ce^{\wt{A}+\wt{B}}} \\
\lambda\mul A&:=\frac{\ce^{\lambda\wt{A}}}{\Tr \ce^{\lambda\wt{A}}}\\
\scal{A,B}_{\circ}&:=\frac{1}{n}\Tr(\wt{A}\wt{B})-\frac{1}{n^{2}}(\Tr\wt{A})(\Tr\wt{B}).
\end{align*}
\end{definition}

First note about these operations that $A\add B,\lambda\mul A\in\MNN{n}$ if $A,B\in\MNp{n}$ and $\lambda\in\R$,
  moreover, they are defined on the rays of positive definite operators
  in the following sense.

\begin{lemma}
\label{th:scaleinv}
For every $A,B\in\MNp{n}$ and $c\in\Rp$ we have the following scale invariance.
\begin{align*}
A\add B&=(cA)\add B=A\add (cB)\\
\lambda\mul A&=\lambda\mul (cA)\\
\scal{A,B}_{\circ}&=\scal{cA,B}_{\circ}=\scal{A,cB}_{\circ}
\end{align*}
\end{lemma}
\begin{proof}
Elementary computation.
\end{proof}

Because of scale invariance one can factorize by rays, which means that is enough to consider only the set of trace one
  matrices i.e. the set of quantum states.
It turns out that the interior of the state space is a Euclidean space with these operations.
  
\begin{theorem}
The state space $\MNN{n}$ is a real Hilbert space with addition $\add$, 
  null vector $\frac{1}{n}I_{n}$, multiplication $\mul$ and 
  inner product $\scal{\cdot,\cdot}_{\circ}$.
\end{theorem}
\begin{proof}
Not every requirement of Hilbert spaces will be proven, just two of them to give insight to such calculations.
Lemma \ref{th:scaleinv} gives the idea to define temporarily the relation $a\simeq b$ 
  if there is a $c\in\Rp$ such that $a=cb$.
  
First let us prove the distributivity for $A,B\in\MNN{n}$ and $\lambda\in\R$.
\begin{align*}
 \lambda\mul(A\oplus B)\simeq  \lambda\mul\ce^{\wt{A}+\wt{B}}
 \simeq \ce^{\lambda(\wt{A}+\wt{B})}
 =\ce^{\lambda \wt{A}+\lambda\wt{B}}
 \simeq \ce^{\lambda \wt{A}}\oplus\ce^{\lambda \wt{B}}
 \simeq (\lambda\mul A)\oplus(\lambda\mul B)
\end{align*}

Now let us prove the additivity of the scalar product for for $A,B,C\in\MNN{n}$.
\begin{align*}
\scal{A\oplus B,C}_{\circ}&=\scal{\ce^{\wt{A}+\wt{B}},C}_{\circ}
=\frac{1}{n}\Tr((\wt{A}+\wt{B})\wt{C})-\frac{1}{n^{2}}(\Tr(\wt{A}+\wt{B})))(\Tr\wt{C})
=\scal{A,C}_{\circ}+\scal{B,C}_{\circ}.
\end{align*}
\end{proof}

Because of the group property of the addition, every state $A$ has an additive inverse, which will be denoted by
  $\ominus A$.
It can be expressed by multiplication as $\ominus A=(-1)\mul A$.

As in the classical case the intervals $(t\mul A)\oplus((1-t)\mul B)$ $(t\in\sz{0,1})$
  has special interests in quantum hypothesis testing \cite{CodingAudenaert,CodingMosonyi,PetzQinf}.
They can be considered as one generalization of Hellinger arc.  

The introduced operations are unitary invariant, which means that they preserve symmetries in quantum mechanics.

\begin{lemma}
\label{th:unitinv}
For every state $A,B\in\MNN{n}$, $\lambda\in\R$ and unitary matrix $U\in\MN{n}$ we have the identities
\begin{align*}
(UAU^{*})\oplus(UBU^{*})&=U(A\oplus B)U^{*}\\
\lambda\mul(UAU^{*})&=U(\lambda\mul A)U^{*}\\
\scal{UAU^{*},UBU^{*}}_{\circ}&=\scal{A,B}_{\circ}.
\end{align*}
\end{lemma}
\begin{proof}
Simple application of the formula $\wt{UAU^{*}}=U\wt{A}U^{*}$.
\end{proof}

The tensor product plays a key role in quantum information theory \cite{NielsenChuang,PetzQinf} since
  it describes composite systems.
The following theorem shows that the introduced scalar product acts on composite systems as they were 
  direct summand of subspaces.

\begin{theorem}
\label{th:scaltensor}
For states $A_{1},A_{2}\in\MNN{n}$ and $B_{1},B_{2}\in\MNN{m}$ the
  scalar product of their tensor product can be expressed as
\begin{equation*}
\scal{A_{1}\otimes B_{1},A_{2}\otimes B_{2}}_{\circ}
 =\scal{A_{1},A_{2}}_{\circ}+\scal{B_{1},B_{2}}_{\circ}.
\end{equation*}
\end{theorem}
\begin{proof}
Elementary computation using Lemma \ref{th:tensorlog}.
\end{proof}

An immediate consequence of the previous theorem is the Pythagorean theorem for composite quantum systems.  

\begin{theorem}
For states $A\in\MNN{n}$ and $B\in\MNN{m}$ we have
\begin{equation*}
\norm{A\otimes B}_{\circ}^{2}=\norm{A}_{\circ}^{2}+\norm{B}_{\circ}^{2}.
\end{equation*}
\end{theorem}
 
The relative entropy is a key concept in quantum information theory \cite{NielsenChuang,PetzQinf},
  which can be defined via modular operators for states $D_{1},D_{2}\in\MNN{n}$ as
\begin{equation*}
S(D_{1},D_{2})=-\scal{D_{1}^{1/2},\wt{\Delta_{D_{2}/D_{1}}}D_{1}^{1/2}}_{n}.
\end{equation*}
This formula is nothing else but Araki's definition of the relative entropy in a general von Neumann algebra \cite{Araki}.
Later Petz used the concept of modular operators to define quasi entropies \cite{PetzQuasi} and
  in 2002 Hiai and Petz presented a unified approach to quantum relative entropies, monotone metrics and quasi entropies
  via modular operators \cite{PetzHiai}.
Note that the logarithm of the modular operator occurs in the formula.
It turns out that the scalar product $\scal{\cdot,\cdot}_{\circ}$ can be expressed in terms of the logarithm of 
  modular operators.

\begin{theorem}
For states $A,B\in\MNN{n}$ the scalar product can be written as
\begin{equation*}
\scal{A,B}_{\circ}=\frac{1}{2n^{2}}\scal{ \wt{\Delta_{A}},\wt{\Delta_{B} }}.
\end{equation*}
\end{theorem}
\begin{proof}
The theorem is equivalent to the equality
\begin{equation}
\label{eq:scalwithmodular}
\scal{A,B}_{\circ}=\frac{1}{2n^{2}}\Tr\gz{\gz{\wt{A\otimes A^{-1}}} \gz{\wt{B\otimes B^{-1}}}},
\end{equation}
which can proved using Lemma \ref{th:tensorlog}.
\end{proof}

It is worth noting the remarkable similarities of the Equations (\ref{eq:oscalCl}) and (\ref{eq:scalwithmodular}).

\subsection{Aitchison geometry on the space of qubits}

As an application we present the Aitchison geometry on the space of qubits $\MNN{2}$ via explicit formulae.
The state space $\MNN{2}$ can be identified with the interior of the unit ball in the three dimensional Euclidean
  space with the map
\begin{equation*}
 \kz{(x,y,z)\in\R^{3}\vert\ \norm{(x,y,z)}<1}\to\MNN{2}\qquad
   (x,y,z)\mapsto\frac{1}{2}\begin{pmatrix}1+z & x+\ci y\\ x-\ci y & 1-z\end{pmatrix}.
\end{equation*}
This unit ball is called Bloch ball and its elements are referred to as quantum states.
The origin of the Bloch ball is the origin in the Aitchison geometry.
The Hilbert space operations on the space $\MNN{2}$ is summarized in the following lemma.

\begin{lemma}
\label{th:qubit} 
Assume that $D_{1}$ and $D_{2}$ are elements in the Bloch ball, such that their Euclidean distance from the
  origin is $R,r$ and the angle between them is $\vartheta$.
We have for their scalar product
\begin{equation}
\label{eq:qubitscal}
 \scal{D_{1},D_{2}}_{\circ}=\artanh(r)\artanh(R)\cos(\vartheta),
\end{equation}
  the length of the vectors are
\begin{equation}
\label{eq:qubitnorm}
 \norm{D_{1}}_{\circ}=\artanh(R), \quad \norm{D_{2}}_{\circ}=\artanh(r),
\end{equation}
  so the angle between them in the Aitchison geometry is $\vartheta$ too.
The distance square between them is
\begin{equation}
\label{eq:qubitdistance}
 \norm{D_{1}\ominus D_{2}}_{\circ}^{2}=\artanh^{2}(R)+\artanh^{2}(r)-2\cos(\vartheta)\artanh(R)\artanh(r).
\end{equation}
The additive inverse is
\begin{equation}
\label{eq:qubitinverse}
 \ominus D_{1}=I_{2}-D_{1},
\end{equation}
  that is the reflection to the origin in the Bloch ball.
The multiplication is a dilatation, for $\lambda\in\Rp$ the Euclidean distance the state $\lambda\mul D_{1}$ from
  the origin in the Bloch ball is $\tanh(\lambda\artanh(R))$.
\end{lemma}
\begin{proof}
The unitary invariance of the metric implies that the norm of a state in the Bloch ball depends
  just on the distance from the origin. 
To simplify the calculations we can assume that the first state lays on the positive part 
  of the axis $z$ and the second state has no $y$ component.
Using the parameterization for states
\begin{equation*}
 D_{1}(R)=\frac{1}{2}\begin{pmatrix}1+R&0\\ 0&1-R\end{pmatrix}
 \qquad
 D_{2}(r,\vartheta)=\frac{1}{2}
   \begin{pmatrix}1+r\cos\vartheta & r\sin\vartheta\\ r\sin\vartheta & 1-r\cos\vartheta\end{pmatrix}
 \end{equation*}
  calculations gives us Equations (\ref{eq:qubitscal},\ref{eq:qubitnorm},\ref{eq:qubitdistance},\ref{eq:qubitinverse}).
\end{proof}  

Finally an orthonormal basis $D_{1},D_{2},D_{3}$ is presented in $\MNN{2}$ with respect to the scalar product
  $\scal{\cdot,\cdot}_{\circ}$.
\begin{equation*}
D_{1}=\frac{1}{2}\begin{pmatrix}1&\tanh 1\\ \tanh1 &1\end{pmatrix}\qquad
D_{2}=\frac{1}{2}\begin{pmatrix}1&\ci\tanh 1\\ -\ci\tanh1 &1\end{pmatrix}\qquad
D_{3}=\frac{1}{\ce+\ce^{-1}}\begin{pmatrix}\ce&0\\ 0 &\ce^{-1}\end{pmatrix}
\end{equation*}

\section{Connection to Gibbs states}

\begin{definition}
In quantum mechanical setting an operator $H\in\MNsa{n}$ is called Hamiltonian and
  for a parameter $\beta\in\Rp$ the state
\begin{equation*}
D=\frac{\ce^{-\beta H}}{\Tr\ce^{-\beta H}}
\end{equation*}
is called Gibbs state of $H$ at inverse temperature $\beta$, or shortly Gibbs state.
\end{definition}

Note that the Hamilton operators $H$ and $H+cI_{n}$ determine the same Gibbs state, so one can assume that the
  Hamiltonian is traceless.
Using the definition of Gibbs states one can assign a state to every Hamiltonian. 
The counterpart of this mechanism in classical probability theory is called softmax function.
Now we define the inverse of this map.

\begin{definition}
 The centered log-ratio transformation is defined as
 \begin{equation*}
  \clr: \MNN{n}\to\TMN{n}\qquad D\mapsto \Tilde{D}-\frac{1}{n}I_{n}\Tr\Tilde{D}.
 \end{equation*}
\end{definition}

Below we show that the centered log-ratio transformation is scale invariant, linear, isometric and 
  preserve tensor product in some sense.
  
\begin{theorem}
\label{th:clrproperties}
For states $A,B\in\MNN{n}$, $C\in\MNN{m}$, scalars $\lambda\in\R$, $c\in\Rp$ 
  and unitary operator $U\in\MN{n}$ we have
\begin{align*}
\clr(cA)&=\clr A\\
\clr(\lambda\mul A)&=\lambda\clr A\\
\clr(A\add B)&=\clr A+\clr B\\
\scal{A,B}_{\circ}&=\frac{1}{n}\scal{\clr A,\clr B}_{n}\\
\clr(A\otimes C)&=\clr A\otimes I_{m}+I_{n}\otimes\clr C\\
\clr(UAU^{*})&=U(\clr A)U^{*}.
\end{align*}
\end{theorem}
\begin{proof}
Simple application of Lemmas \ref{th:tensorlog}, \ref{th:scaleinv} and \ref{th:unitinv}.
\end{proof}

This theorem gives an interpretation of the vector operations in the state space.
The sum of Gibbs states $D_{1}\oplus D_{2}$ generated by Hamiltonians $H_{1}$ and $H_{2}$ at the same
  inverse temperature $\beta$ is just the Gibbs state corresponding to the Hamiltonian $H_{1}+H_{2}$. 
The state $\lambda\mul D_{1}$ ($\lambda\in\R$) can be interpreted as a Gibbs state generated by either 
  $\lambda H_{1}$ at $\beta$ or $H_{1}$ at $\lambda\beta$.
Finally, $\scal{D_{1},D_{2}}$ is just the Hilbert-Schmidt inner product of the Hamiltonians
  $\scal{H_{1},H_{2}}$ up to a normalizing factor.
  
The above mentioned properties of the $\clr$ function indicate a simple proof for Theorem \ref{th:scaltensor}.
For states $A_{1},A_{2}\in\MNN{n}$ and $B_{1},B_{2}\in\MNN{m}$ the calculations below prove the theorem.
\begin{align*}
\scal{A_{1}\otimes B_{1},A_{2}\otimes B_{2}}_{\circ}
 &=\frac{1}{nm}\scal{\clr(A_{1}\otimes B_{1}),\clr(A_{2}\otimes B_{2})}_{n} \\
&=\frac{1}{nm}\scal{\clr A_{1}\otimes I_{m}+I_{n}\otimes\clr{B_{1}},
   \clr A_{2}\otimes I_{m}+I_{n}\otimes\clr{B_{2}} }_{n} \\
&=\frac{1}{nm}\bigl(\Tr((\clr{A_{1}}\clr{A_{2})}\otimes I_{m})+\Tr(\clr{A_{1}\otimes \clr B_{2}})\\
&\qquad  +\Tr(\clr A_{2}\otimes\clr B_{1})+\Tr(I_{n}\otimes(\clr B_{1}\clr B_{2})) \bigr)\\
&=\frac{1}{n}\Tr(\clr A_{1}\clr A_{2})+\frac{1}{nm}\Tr(\clr A_{1})\Tr(\clr B_{2})\\
&\qquad  +\frac{1}{nm}\Tr(\clr A_{2})\Tr(\clr B_{1})+\frac{1}{m}\Tr(\clr B_{1}\clr B_{2})\\
&=\frac{1}{n}\scal{\clr(A_{1}),\clr(A_{2})}_{n}+\frac{1}{m}\scal{\clr(B_{1}),\clr(B_{2})}_{n}\\
&=\scal{A_{1},A_{2}}_{\circ}+\scal{B_{1},B_{2}}_{\circ}
\end{align*}

Aitchison et al. gave an orthonormal basis on the probability simplex in 2002 \cite{AitchisonONB}. 
Here we present the noncommutative analogue of their basis.

\begin{theorem}
For a given $n$ define $a=\sqrt{\frac{n}{2}}$ and the following matrices.
\begin{align*}
&1\leq k<l\leq n:&\quad
  A_{kl}&=\frac{1}{n-2+2\cosh a}\gz{I_{n}+(\cosh a-1)(E_{kk}+E_{ll})+(\sinh a)(E_{kl}+E_{lk})}\\
&1\leq k<l\leq n:&\quad
  B_{kl}&=\frac{1}{n-2+2\cosh a}\gz{I_{n}+(\cosh a-1)(E_{kk}+E_{ll})+(\ci\sinh a)(E_{kl}-E_{lk})}\\
&1\leq k\leq n-2:&\quad
  C_{k}&=\frac{1}{(k+1)\ce^{\alpha}+\ce^{-(k+1)\alpha}+k-2 }
\Diag\gz{\underbrace{\ce^{\alpha},\dots,\ce^{\alpha}}_{k},\ce^{-(k+1)\alpha},\underbrace{1,\dots,1}_{n-k-2},\ce^{\alpha}}\\
& &\quad  &\mbox{where}\quad \alpha=\sqrt{\frac{n}{k^{2}+3k+2}}\\
&k=n-1:&\quad 
   C_{n-1}&=\frac{1}{n-a+2\cosh a}\gz{I_{n}+(\ce^{a}-1)E_{11}+(\ce^{-a}-1)E_{nn}}
\end{align*}
The set of matrices $\kz{A_{kl},B_{kl}}_{1\leq k<l\leq n}\cup\kz{C_{k}}_{1\leq k\leq n-1}$ form an orthonormal basis
  in $\MNN{n}$ with respect to the scalar product $\scal{\cdot,\cdot}_{\circ}$.
\end{theorem}
\begin{proof}
 The matrices $A_{kl},B_{kl}$ and $C_{k}$ are the Gibbs states generated by Hamiltonians
\begin{align*}
&1\leq k<l\leq n;&\quad
  \mathcal{A}_{kl}&=a(E_{kl}+E_{lk})\\
&1\leq k<l\leq n;&\quad
  \mathcal{B}_{kl}&=\ci a(E_{kl}-E_{lk})\\
&1\leq k\leq n-2:&\quad
  \mathcal{C}_{k}&=\sqrt{\frac{n}{k^{2}+3k+2}}
\Diag\gz{\underbrace{1,\dots,1}_{k},-k-1,\underbrace{0,\dots,0}_{n-k-2},1}\\
&k=n-1:&\quad 
   \mathcal{C}_{n-1}&=a(E_{11}-E_{nn}).
\end{align*}
Simple computation shows that these matrices form an orthonormal basis in $\TMN{n}$ with respect 
  to the normalized Hilbert-Schmidt scalar product $\frac{1}{n}\scal{\cdot,\cdot}_{n}$.
According to Theorem \ref{th:clrproperties} the corresponding Gibbs states form an orthonormal basis
  with respect to $\scal{\cdot,\cdot}_{\circ}$ in $\MNN{n}$.
\end{proof}

\section{Conclusions}

In this paper the Aitchison geometry was generalized from classical statistical models to the
  quantum mechanical finite dimensional state space, that is a Hilbert-space structure on the state space was presented.
The intervals in this space was already used in quantum hypothesis testing, as in the classical case.
The introduced scalar product, inspired by the classical one, surprisingly preserve the tensor product 
  and can be expressed in terms of modular operators and Hamilton operators.
The Hilbert space structure of the state space is turned out to be unitary invariant.
The analogue of the log-ratio transformation was given in this quantum setting which helped give an
  orthonormal basis in the state space.
The above mentioned notions were studied in detail on the space of qubits.

As in the classical case, where the Aitchison geometry first was presented on the simplex and later
  extended to more and more complicated statistical models, in quantum setting, the presented geometry 
  is just the first step in this direction.
In the future the generalization to infinite dimensional state spaces case 
  could be the counterpart of the classical continuous case.
The generalization to the Radon-Nykodim derivative of states in a von Neumann algebra could be the considered
  as a noncommutative version of the existing extension to Radon-Nykodim derivative of measures.

\section*{Acknowledgments}

Both of the authors enjoyed the support of the Hungarian National Research,
  Development and Innovation Office (NKFIH) grant no. K124152.


\bibliography{ref} 
\bibliographystyle{plain}

\end{document}